\newif\ifstoc
\stoctrue 
\stocfalse

\newcommand{\stocoption}[2]
{\ifstoc%
#1%
\else%
#2%
\fi}

\ifstoc
\documentclass{llncs}
\else
\documentclass[11pt]{article}

\usepackage{typearea}
\typearea{14}

\usepackage[compact]{titlesec}

\fi

\usepackage{epsfig}
\usepackage{amsfonts}
\usepackage{amssymb}
\usepackage{amstext}
\usepackage{amsmath}
\usepackage{xspace}
\usepackage{algorithm}
\usepackage[noend]{algorithmic}

\allowdisplaybreaks

\ifstoc

\newtheorem{cl}[theorem]{Claim}
\newtheorem{fact}{Fact}
\else

\newtheorem{theorem}{Theorem}[section]
\newtheorem{definition}[theorem]{Definition}
\newtheorem{lemma}[theorem]{Lemma}

\newtheorem{cl}[theorem]{Claim}

\newenvironment{proof}{{\bf Proof:  }}{\hfill\rule{2mm}{2mm}}

\fi

\numberwithin{algorithm}{section}

\newcommand{\ignore}[1]{}

\newcommand{\initOneLiners}{%
    \setlength{\itemsep}{0pt}
    \setlength{\parsep }{0pt}
    \setlength{\topsep }{0pt}
}

\newenvironment{OneLiners}[1][\ensuremath{\bullet}]
    {\begin{list}
        {#1}
        {\initOneLiners}}
    {\end{list}}

\newcommand{\sse}{\subseteq}
\def\opt{{\sf opt}\xspace}
\def\alg{{\sf alg}\xspace}
\def\pin{\ensuremath{{\cal I}_{in}}\xspace}
\def\kin{\ensuremath{{k}_{in}}\xspace}
\def\pot{\ensuremath{{\cal I}_{out}}\xspace}
\def\kot{\ensuremath{{k}_{out}}\xspace}
\def\lp{\ensuremath{{\cal LP}}\xspace}
\def\p{\ensuremath{{\cal P}}\xspace}
\def\cL{\ensuremath{{\cal L}}\xspace}
\def\E{\ensuremath{{\mathbb E}}}

\newcommand{\I}{\ensuremath{\mathcal{I}}\xspace}
\newcommand{\J}{\ensuremath{\mathcal{J}}\xspace}
\newcommand{\K}{\ensuremath{\mathcal{K}}\xspace}
\newcommand{\D}{\ensuremath{\mathcal{D}}\xspace}
\renewcommand{\P}{\ensuremath{\mathcal{P}}\xspace}

\DeclareMathOperator{\spn}{span}

\newcounter{mynote}[section]
\renewcommand{\themynote}{\thesection.\arabic{mynote}}
\stocoption{
\newcommand{\agnote}[1]{}
\newcommand{\vnote}[1]{}

}{
\newcommand{\agnote}[1]{\refstepcounter{mynote}$\ll${\bf Anupam~\themynote:} {\sf #1}$\gg$\marginpar{\tiny\bf AG~\themynote}}
\newcommand{\vnote}[1]{\refstepcounter{mynote}$\ll${\bf Viswa~\themynote:} {\sf #1}$\gg$\marginpar{\tiny\bf VN~\themynote}}
}

\title{A Stochastic Probing Problem with Applications} 
\date{}

\ifstoc
\author{ Anupam Gupta\inst{1}
 \and Viswanath Nagarajan\inst{2}
}

\institute{Computer Science Department, Carnegie Mellon University. \and IBM T.J. Watson Research Center. }
\else
\author{
Anupam Gupta\thanks{Computer Science Department, Carnegie Mellon University, Pittsburgh, PA 15213, USA. Supported in part by NSF award CCF-1016799  and an Alfred P.~Sloan Fellowship.}
 \and Viswanath Nagarajan\thanks{IBM T.J. Watson Research Center. }
}
\fi

\begin{document}
\maketitle

\begin{abstract}
  We study a general {\em stochastic probing} problem defined on a
  universe $V$, where each element $e\in V$ is ``active'' independently
  with probability $p_e$. Elements have weights $\{w_e:e\in V\}$ and the goal is
  to maximize the weight of a chosen subset $S$ of active elements.
  However, we are given only the $p_e$ values---to determine whether or
  not an element $e$ is active, our algorithm must {\em probe}~$e$. If
  element $e$ is probed and happens to be active, then $e$ must
  irrevocably be added to the chosen set $S$; if $e$ is not active then
  it is not included in $S$. Moreover, the following conditions must
  hold in every random instantiation:
  \begin{itemize}
  \item the set $Q$ of probed elements satisfy an ``outer'' packing
    constraint, 
  \item the set $S$ of chosen elements satisfy an ``inner'' packing
    constraint.
  \end{itemize}
  The kinds of packing constraints we consider are intersections of
  matroids and knapsacks. 
  Our results provide a simple and unified view of results in stochastic  matching~\cite{CIKMR09,BGLMNR12} and Bayesian mechanism design~\cite{CHMS10}, and can also handle more general constraints. As an application, we obtain the first polynomial-time $\Omega(1/k)$-approximate ``Sequential Posted  Price Mechanism'' 
under $k$-matroid intersection feasibility constraints, improving on prior work~\cite{CHMS10,Y11,KW12}.    
\end{abstract}

\section{Introduction}

We study an adaptive stochastic optimization problem along the lines
of~\cite{MSU99,DGV04,DGV05,GM07}. The {\em stochastic probing} problem is
defined on a universe $V$ of elements with weights $\{w_e:e\in V\}$.
We are also given two downwards-closed set systems $(V, \pin)$ and $(V,
\pot)$, which we call the \emph{inner} and \emph{outer} packing
constraints, whose meanings we shall give shortly. For each element $e
\in V$, there is a probability $p_e$, where element $e$ is
\emph{active}/\emph{present} with this probability, independently of all
other elements. We want to choose a set $S \subseteq V$ of active elements belonging to
$\pin$, i.e., all elements in the chosen set $S$ must be
active and also independent according to the inner packing constraint
($S \in \pin$). The goal is to maximize the expected weight of the
chosen set.

However, the information about which elements are active and which are
inactive is not given up-front. All we know are the probabilities $p_e$,
and that the active set is a draw from the product distribution given by
$\{p_e\}_{e \in V}$---to determine if an element $e$ is active or not,
we must {\em probe} $e$. Moreover, if we probe $e$, and $e$ happens to
be active, then we {\em must} irrevocably add $e$ to our chosen set
$S$---we do not have a right to discard any probed element that turns
out to be active. This ``query and commit'' model is quite natural in a
number of applications such as kidney exchange, online dating and
auction design (see below for details).

Finally, there is a constraint on which elements we can probe: the set
$Q$ of elements probed in any run of the algorithm must be independent
according to the outer packing constraint $\pot$---i.e., $Q \in \pot$.
This is the constraint that gives the probing problem its richness.
Since every probed element that is active must be included in the
solution which needs to maintain independence in $\pin$, at any point
$t$ (with current solution $S_t$ and currently probed set $Q_t$) we can
only probe those elements $e$ with $Q_t
\cup \{e\} \in \pot$ and $S_t \cup \{e\} \in \pin$.\stocoption{\footnote{Indeed, if $p_e = 0$ there is no point
  probing $e$; and if $p_e>0$ there is a danger that $e$ is active and
  we will be forced to add it to $S_t$, which we cannot if $S_t \cup
  \{e\} \not\in \pin$.}}{(Indeed, if $p_e = 0$ there is no point
  probing $e$; and if $p_e>0$ there is a danger that $e$ is active and
  we will be forced to add it to $S_t$, which we cannot if $S_t \cup
  \{e\} \not\in \pin$.)}

While the stochastic probing problem seems fairly abstract, it has
interesting applications: we give two applications of this problem, to
designing posted-price Bayesian auctions, and to modeling problems in
online dating/kidney exchange. We first state our results 
and then describe these applications. 

\stocoption{
\smallskip
{\bf Our Results.}
}{
\subsection{Our Results}
\label{sec:our-result}
}
For the {\em unweighted} stochastic probing problem (i.e., $w_e = 1$ for
all $e \in V$), if both inner and outer packing constraints are given by
$k$-systems\footnote{For any integer $k$, a \emph{$k$-system} is a
  downwards-closed collection of sets $\I \sse 2^V$ such that for any $S
  \sse V$, the maximal subsets of $S$ that belong to $\I$ can differ in
  size by at most a factor of $k$. Examples 
  are intersections of $k$ matroids, and $k$-set packing.}, we consider
the greedy algorithm which considers elements in decreasing order of
their probability $p_e$, probing them whenever feasible.
\begin{theorem}[Unweighted Probing]
  \label{thm:unit-wt}
  The greedy algorithm for unweighted stochastic probing achieves a
  tight $\frac1{\kin + \kot}$-approximation ratio, when $\pin$ is a
  $\kin$-system and $\pot$ is a $\kot$-system.
\end{theorem}
This result generalizes the greedy 4-approximation algorithm for
unweighted stochastic matching, by Chen et al.~\cite{CIKMR09}, where
both inner and outer constraints are $b$-matchings (and hence
$2$-systems). For the special case of stochastic matching,
Adamczyk~\cite{A11} gave an improved factor-$2$ bound. However,
Theorem~\ref{thm:unit-wt} is tight in our setting of general
$k$-systems; its proof is LP-based, and we feel it is much simpler than
previous proofs for the special cases. The main idea of our proof
is a dual-fitting argument that extends the Fisher et al.~\cite{FNW78}
analysis of the greedy algorithm for $k$-matroid intersection. 
\stocoption{}{
In Section~\ref{sec:unwt-deadline} we generalize our unweighted probing result to also handle ``global time'' constraints, at the loss of a small constant factor in the approximation ratio (see Theorem~\ref{thm:unwt-deadline}). }

There is no known greedy algorithm for stochastic probing in the
\emph{weighted} case (as opposed to the deterministic setting of finding
the maximum weight set subject to a $k$-system, where greedy gives a
$1/k$-approximation~\cite{J76,FNW78}); indeed, natural greedy approaches
can be arbitrarily bad even for weighted stochastic
matching~\cite{CIKMR09}. Hence, we use an LP relaxation for the weighted
probing problem, where variables correspond to marginal probabilities of
probing/choosing elements in the optimal policy. This is similar to
previous works on such adaptive stochastic problems~\cite{DGV04,DGV05,BGLMNR12}. Our rounding algorithm is based on the recently introduced
notion of {\em contention resolution (CR) schemes} for packing
constraints, due to Chekuri et al.~\cite{CVZ11}. \stocoption{}{Loosely speaking, given
a packing constraint on a universe $V$ and a fractional solution
$\{x_e\}_{e\in V}$, a CR-scheme is a two-step rounding procedure where
\begin{OneLiners}
\item[a.] Each element $e$ is chosen independently into $I_1\sse V$ with
  probability proportional to $x_e$.
\item[b.] A feasible subset $I_2\sse I_1$ (suitably computed) is output
  as the solution.  
\end{OneLiners} } We show that the existence of suitable CR-schemes for both $\pin$ and
$\pot$ imply an approximation algorithm for weighted stochastic probing,
where the approximation ratio depends on the quality of the two
CR-schemes. Our main result for weighted stochastic probing is
Theorem~\ref{thm:stoc-probe} (which requires some notation to state
precisely), \stocoption{ but a representative corollary is an $\Omega\left(\frac{1}{\kin+\kot}\right)$-approximation
  algorithm 
when the inner and outer constraints  are intersections of $\kin$ and $\kot$ matroids, respectively.}{but here is a representative corollary:
\begin{theorem}[Weighted Probing: Special Case]
  \label{thm:main2-specialcase}
  There is an $\Omega\left(\frac{1}{\kin+\kot}\right)$-approximation
  algorithm for weighted stochastic probing when the inner and outer constraints
  are intersections of $\kin$ and $\kot$ matroids, respectively. Moreover, there is an $\Omega\left(\frac{1}{(\kin+\kot)^2}\right)$-approximation algorithm under arbitrary $\kin$ and $\kot$ system constraints.
\end{theorem}
} Some of the other allowed constraints are unsplittable flow on trees
(under the ``no-bottleneck'' assumption) and packing integer programs.
Details on the weighted case appear in Section~\ref{sec:weighted}.

\stocoption{
\smallskip
{\bf Applications.}
}{
\subsection{Applications} \label{sec:applications}
}
We now give two applications: the first shows how our algorithm for the
weighted probing problem immediately gives us posted price auctions for single
parameter settings where the feasibility set is given by intersections
of matroids, the second is an application for dating/kidney exchange.
Both of these extend and generalize previous results in these areas.

{\bf Bayesian Auction Design.} Consider a mechanism design setting
for a single seller facing $n$ single-parameter buyers.  The seller has
a feasibility constraint given by a downward-closed set system $\I\sse
2^{[n]}$ and is allowed to serve any set of buyers from $\I$. Buyers are
{\em single-parameter}; i.e., buyer $i$'s private data is a single real
number $v_i$ which denotes his valuation of being served (if $i$ is not
served then he receives zero value). In the {\em Bayesian} setting, the
valuation $v_i$ is drawn from some set $\{0, 1, \ldots, B\}$ according
to probability distribution $\D_i$; here we assume that the valuations
of buyers are discrete and independently drawn.  The valuation $v_i$ is private to
the buyer, but the distribution $\D_i$ is public knowledge. The goal in
these problems is 
a revenue-maximizing truthful mechanism that accepts bids from
buyers and outputs a feasible allocation (i.e., a set $S \in \I$ of
buyers that receive service), along with a price that each buyer has to
pay for service.  A very special type of mechanism is a {\em Sequential
  Posted Pricing Mechanism} (SPM) that chooses a price for each buyer
and makes ``take-it-or-leave-it'' offers to the buyers in some
order~\cite{SG06,BH08,CHMS10}. Such mechanisms are simple to run and
obviously truthful (see~\cite{CHMS10} for a discussion of other
advantages), hence it is of interest to design SPMs which achieve
revenue comparable to the revenue-optimal mechanism.

Designing the best SPM can be cast as a stochastic probing problem on a
universe $V = \{1,2,\ldots, n\} \times \{0,1,\ldots, B\}$, where element
$(i,c)$ corresponds to offering a price $c$ to buyer $i$. Element
$(i,c)$ has weight $w_{ic} = c$, which is the revenue obtained if the
offer ``price $c$ for buyer $i$'' is accepted, and has probability
$p_{ic} = \Pr_{v_i \sim \D_i}\left[ v_{i} \geq c \right]$, which is the
probability that $i$ will indeed accept service at price $c$. The inner
constraint $\pin$ is now the natural lifting of the actual constraints
$\I$ to the universe $V$, where $\{(i,c)\}_{c\ge 0}$ are copies of $i$.
The outer constraint $\pot$ requires that at most one of the elements
$\{(i,c) \mid c \geq 0\}$ can be probed for each $i$: i.e., each buyer
$i$ can be offered at most one price. This serves two purposes: firstly,
it gives us a posted-price mechanism. Secondly, we required in our model
that each element $(i,c)$ is active with probability $p_{ic}$, {\em
  independently} of the other elements $(i,c')$; however, the underlying
semantics imply that if $i$ accepts price $c$, then she would also
accept any $c'\le c$, which would give us correlations. Constraining
ourselves to probe at most one element corresponding to each buyer $i$
means we never probe two correlated elements, and hence the issue of
correlations never arises.

Our results for stochastic probing give near-optimal SPMs for many
feasibility constraints. Moreover, we show that our LP relaxation not
only captures the best possible SPMs, but also captures the optimal
truthful mechanism \emph{of any form} under the Bayes-Nash equilibrium
(and hence Myerson's optimal mechanism~\cite{M81}).
In the case of $k$ matroid intersection feasibility constraints, our
results give the first polynomial-time sequential posted price
mechanisms whose revenue is $\Omega(1/k)$ times the optimum. Previous
papers~\cite{CHMS10,Y11,KW12} proved the existence of such SPMs, but
they were polynomial-time only for $k\le 2$. For larger $k$, previous
works only showed \emph{existence} of $\Omega(1/k)$-approximate SPMs, and polynomial-time
implementations of these SPMs only obtained an $\Omega(1/k^2)$ fraction
of the optimal revenue. The
previous results compare the performance of their SPMs directly to the
revenue of the optimal mechanism~\cite{M81}, whereas we compare our SPMs
to an LP relaxation of this mechanism, which is potentially larger. Moreover, our general framework gives us more power:
\begin{itemize}
\item We can handle broader classes of feasibility constraints~$\I$, not
  just matroid intersections: e.g., we can model auctions involving
  unsplittable flow on trees, which can be used to capture allocations
  of point-to-point bandwidths in a tree-shaped network. This is because
  the feasibility constraints $\I$ for the auction directly translate
  into inner constraints for the probing problem.

\item We can also handle additional side-constraints to the
  auction via a richer class of outer constraints $\pot$. For example,
  the seller may incur costs in the form of time/money to make offers.
  Such budget limits can be modeled in the stochastic probing problem as
  an extra outer knapsack constraint, and our algorithm finds
  approximately optimal SPMs even in this case. More generally, our
  algorithm can easily handle a rich class of other resource constraints
  (matroid intersections, packing IPs etc) on the auction. However,
  in the presence of these side-constraints, our algorithm's
  revenue is an approximation only to the best SPM satisfying these
  constraints, and no longer comparable to the unconstrained optimal
  mechanism.
\end{itemize}

{\bf Online dating and Kidney Exchange~\cite{CIKMR09}} Consider a
dating agency with several users. Based on the profiles of users, the
agency can compute the probability that any pair of users will be
compatible. Whether or not a pair is successfully matched is only known
after their date; moreover, in the case of a match, both users
immediately leave the site (happily). Furthermore, each user has a
patience/timeout level, which is the maximum number of failed dates after which
he/she drops out of the site (unhappily). The objective of the dating
site is to schedule dates so as to maximize the expected number of
matched pairs. (Similar constraints arise in kidney exchange systems.)
This can be modeled as stochastic probing with the universe $V$ being
edges of the complete graph whose nodes correspond to users. The inner
constraints specify that the chosen edges be a matching in $G$.  The
outer constraints specify that for each node $j$, at most $t_j$ edges
incident to $j$ can be probed, where $t_j$ denotes the patience level of
user $j$. Both these are $b$-matching constraints; in fact when the
graph is bipartite, they are intersections of two partition matroids.

Our results will give an alternate way to obtain constant factor
approximation algorithms for this stochastic matching problem.  Such
algorithms were previously given by~\cite{CIKMR09,BGLMNR12}, but they
relied heavily on the underlying graph structure. Additionally, our
techniques allow for more general sets of constraints. E.g., not all
potential dates may be equally convenient to a user, and (s)he might
prefer dates with other nearby users. This can be modeled as a sequence
of patience bounds for the user, specifying the maximum number of dates
that the user is willing to go outside her neighborhood/city/state etc.
In particular, if $u_1,u_2,\ldots, u_n$ denote the users in decreasing
distance from user $j$ then
there is a non-decreasing sequence $\langle t_j^1,\ldots,t_j^n\rangle$
of numbers where user $j$ wishes to date at most $t^r_j$ users among the
$r$ farthest other users $\{u_1,\ldots,u_r\}$.  This corresponds to the
stochastic probing problem, where the inner constraint remains matching
but the outer constraint becomes a $2$-system. Our algorithm achieves a
constant approximation even here. 


\stocoption{
\smallskip
{\bf Other Related Work.}
}{
\subsection{Other Related Work} \label{sec:previous-work}
}
 Dean et al.~\cite{DGV04,DGV05} were the first to consider
approximation algorithms for {\em stochastic packing} problems in the
adaptive optimization model. For the stochastic knapsack problem, where
items have random sizes (that instantiate immediately after selection),
\cite{DGV04} gave a $(3+\epsilon)$-approximation algorithm; this was
improved to $2+\epsilon$ in~\cite{BGK11,B11}. \cite{DGV05} considered
stochastic packing integer programs (PIPs) and gave approximation
guarantees matching the best known deterministic bounds. Our stochastic
probing problem can be viewed as a two-level generalization of
stochastic packing, with two different packing constraints: one for probed
elements, and one for chosen elements. However, all random variables in our
setting are $\{0,1\}$-valued (each element is either active or not),
whereas~\cite{DGV04,DGV05} allow arbitrary non-negative random
variables.

Chen et al.~\cite{CIKMR09} first studied a {\em stochastic probing}
problem: they introduced the unweighted stochastic matching problem and
showed that greedy is a $4$-approximation algorithm. Adamczyk~\cite{A11}
improved the analysis to show a bound of $2$. Both these proofs involve
intricate arguments on the optimal decision tree. In contrast, our
analysis of greedy is much simpler and LP-based, and extends to the more
general setting of $k$-systems. (For the stochastic matching, our result
implies a $4$-approximation.) Bansal et al.~\cite{BGLMNR12} gave a
different LP proof that greedy is a $5$-approximation for stochastic
matching, but their proof relied heavily on the graph structure, making
the extension to general $k$-systems unclear.  \cite{BGLMNR12} also gave
the first $O(1)$-approximation for {\em weighted} stochastic matching,
which was LP-based. (\cite{CIKMR09} showed that natural greedy approaches
for weighted stochastic matching are arbitrarily bad.) Our algorithm for
weighted probing is also LP-based, where we make use of the elegant
abstraction of ``contention resolution schemes'' introduced by Chekuri
et al.~\cite{CVZ11} (see Section~\ref{sec:weighted}), which provides a
clean approach to rounding the LP.

The papers of Chawla et al.~\cite{CHMS10}, Yan~\cite{Y11}, and Kleinberg
and Weinberg~\cite{KW12} study the performance of Sequential Posted
Price Mechanisms (SPMs) for Bayesian single-parameter auctions, and
relate the revenue obtained by SPMs to the optimal (non-posted-price)
mechanism given by Myerson~\cite{M81}. Our algorithm for stochastic
probing also yields SPMs for Bayesian auctions where the feasible sets
of buyers are specified by, e.g., $k$-matroid intersection and
unsplittable flow on trees. Our proof relates an LP relaxation of the
optimal mechanism to the LP used for stochastic probing.  Linear
programs have been used to model optimal auctions in a number of
settings; e.g., see Vohra~\cite{V11}.  Bhattacharya et al.~\cite{BGGM10}
also used LP relaxations to obtain approximately optimal mechanisms in a
Bayesian setting with multiple items and budget constrained buyers.

\stocoption{}{
\subsection{Preliminaries} \label{sec:preliminaries}
}
\smallskip
{\bf Specifying Probing Algorithms.} A solution (policy) to the stochastic probing problem is an {\em adaptive}
strategy of probing elements satisfying the constraints imposed by \pot
and \pin.
At any time step $t\ge 1$, let $Q_t$ denote  the set of elements already probed and $S_t$ the current solution (initially $Q_1=S_1=\emptyset$); an element $e\in V\setminus Q_t$ can be probed at time $t$ if and only if $Q_t \cup \{e\}\in \pot$ and $S_t \cup\{e\}\in\pin$.
If $e$ is probed then exactly one of the following happens:
\begin{OneLiners}
\item $e$ is active (with probability $p_e$), and $Q_{t+1} \gets Q_t \cup\{e\}$,
  $S_{t+1} \gets S_t \cup\{e\}$, or
\item $e$ is inactive (with probability $1-p_e$), and $Q_{t+1} \gets
  Q_t \cup\{e\}$, $S_{t+1} \gets S_t $.
\end{OneLiners}
Hence the policy is a decision tree with nodes representing elements that are
probed and branches corresponding to their random instantiations. Note that an optimal policy may
be exponential sized, and designing a polynomial-time algorithm requires
tackling the question of whether there exist poly-sized near-optimal
strategies. A {\em non adaptive} policy is simply given by a permutation on $V$, where elements are considered in this order and probed whenever feasible in both $\pot$ and $\pin$. The \emph{adaptivity gap} compares the best non-adaptive
policy to the best adaptive policy.

{\bf Packing Constraints.} We model packing constraints as {\em independence systems}, which are of
the form $(V,\,\I\sse 2^V)$ where $V$ is the \emph{universe} and $\I$ is
a collection of \emph{independent sets}. We assume $\I$ is
\emph{downwards closed}, i.e., $A\in\I$ and $B\sse A$ $\implies$
 $B\in\I$. Some examples are:
\begin{OneLiners}
\item \emph{Knapsack constraint}: each element $e\in V$ has size
  $s_e\in [0,1]$ and $\I = \{ A\sse V \mid \sum_{e\in A} s_e\le 1 \}$.
\item \emph{Matroid constraint}: an independence system $(V,\I)$ where
  for any subset $S\sse V$, every maximal independent subset of $S$ has
  the same size. See~\cite{Sch-book} for many properties and
  examples.
\item \emph{$k$-system:} an independence system $(V,\I)$ where for any
  subset $S\sse V$, every \emph{maximal} independent subset of $S$ has
  size at least $\frac1k$ times the size of the \emph{maximum}
  independent subset of $S$. For example: matroids are $1$-systems, matchings
  are $2$-systems, and intersections of $k$ matroids form $k$-systems.
\item {\em Unsplittable Flow Problem (UFP) on trees:} there is an  edge-capacitated tree $T$, and each element $e\in V$ corresponds to a path $P_e$ in $T$ and demand $d_e$. Subset $S\sse V$ is independent (i.e. $S\in \I$) iff $\{\mbox{path }P_e \mbox{ with demand }d_e\}_{e\in S}$ is routable in $T$. We assume the ``no-bottleneck'' condition, where the maximum demand $\max_{e\in V} d_e$ is at most the minimum capacity in $T$.
\end{OneLiners}

When the universe is clear from context, we refer to an independence
system $(V,\I)$ just as $\I$. We also make use of linear programming
relaxations for independence systems: the LP relaxation of $\I$ is
denoted by $\P(\I) \sse [0,1]^V$ and contains the convex hull of all
independent sets. (Since $\P(\I)$ is a relaxation it need not equal the
convex hull). For example: $\P(\I)=\left\{\mathbf{x}\in  [0,1]^V : \sum_{e\in V} s_e\cdot x_e\le 1\right\}$ for
knapsacks;  $\P(\I)= \left\{\mathbf{x}\in [0,1]^V : \sum_{e\in S} x_e \le
  r_\I(S),\forall S\sse V\right\}$ for matroids, where $r_\I(\cdot)$ denotes the rank function.

\stocoption{}{
\subsection{Outline}
We first consider the unweighted probing problem in Section~\ref{sec:unwtd}.
Then, in Section~\ref{sec:weighted} we study the weighted probing problem. In Section~\ref{sec:spm} we present the application to posted price mechanisms for Bayesian auctions (Theorem~\ref{thm:spm}). Finally, in Section~\ref{sec:unwt-deadline} we study the generalization of unweighted probing to the setting of global time constraints. 
}

\section{Unweighted Stochastic Probing}
\label{sec:unwtd}

In this section, we study the stochastic probing problem with unit
weights, i.e., $w_e=1$ for all $e\in V$. We assume the inner and outer
packing constraints are a $\kin$-system and a $\kot$-system,
repectively. We show that the greedy algorithm, which considers elements
in non-increasing order of their probabilities $p_e$ and probes them
when feasible, has performance claimed in Theorem~\ref{thm:unit-wt}. 
We give an LP-based dual-fitting proof of this result.

For brevity, let us use $k$ to denote $\kin$, and $k'$ to denote $\kot$.
Let the \emph{rank function} of $\pin$ be $r:2^V\rightarrow
\mathbb{N}$, where for each $S\sse V$, $r(S) = \max\{ |I| \mid I \in \I,
I \sse S\}$ be the \emph{maximum} size of an independent subset of $S$.  By
definition of $k$-systems, for any $S \sse V$, any maximal independent
set of $S$ (according to $\pin$) has size at least $r(S)/k$.  Similarly, let
$r':2^V\rightarrow \mathbb{N}$ denote the rank function of $\pot$.
We may not be able to evaluate the rank function, since this is
NP-complete for $k\ge 3$.
For any $T\sse V$, let $\spn(T) = \{e\in V :
r(T\cup\{e\})=r(T)\}$ be the {\em span} of $T$. Likewise, let $\spn'$
denote the span function for $\pot$. 
\begin{cl}\label{cl:psys-prop}
  For any $T\sse V$, the maximum independent subset of $T$ (which has size $r(T)$) is a maximal independent subset of $\spn(T)$.
  Hence, for $T\sse V$ and $R\sse V$, we have
  $r(\spn(T))\le k\cdot r(T)\le k\cdot |T|$ and $r'(\spn'(R))\le k'\cdot r'(R)\le k'\cdot |R|$.
\end{cl}

Let us write the natural LP relaxation and dual for the probing problem:
{\small \begin{equation*}
\begin{array}{lll|}
  \max& \sum_{e \in V} p_e y_e &\\
  \text{s.t.} &\sum_{e \in S} p_e y_e \leq r(S)&  \,\, \forall
  S\subseteq V \quad \\  
  &\sum_{e \in S} y_e \leq r'(S) & \,\, \forall S \subseteq V \\
  &\mathbf{y} \geq 0.&
\end{array}
\quad 
\begin{array}{lll}
  \min &\sum_S r(S)\,\alpha(S) +  \sum_S r'(S)\,\beta(S)  \\
\text{s.t.}&  p_e \, \sum_{S: e\in S} \alpha(S) + \sum_{S: e \in S} \beta(S) \ge p_e&  \,\, \forall e \in V\\
&  \alpha(S), \beta(S) \geq 0 & \,\, \forall S \sse V.
\end{array}
\end{equation*}
}%
Claim~\ref{cl:lp-opt} in the next section shows that this LP is a valid
relaxation. 
It is not known if these linear programs can be solved in polynomial
time for arbitrary $p$-systems $\pin$ and $\pot$; we use them only for
the analysis. Note that the greedy algorithm defines a non-adaptive
strategy.
Consider a sample path $\pi$ down the natural decision tree associated
with the above algorithm; it is completely defined by the randomness in
which elements are active. Let $\Pr[\pi]$ denote its probability, and
$Q_\pi, S_\pi$ be the sets probed and picked on taking this
path.

\begin{lemma}
  \label{lem:alg-value}
  If $\alg$ is the random variable denoting the number of elements
  picked, 
{\small  \[ \E[\alg] = \sum_{\pi} \Pr(\pi) \cdot |S_\pi| = \sum_{\pi} \Pr(\pi)
  \cdot \sum_{e \in Q_\pi} p_e. \]}
\end{lemma}
\stocoption{}{\begin{proof}
  The first equality follows by definition of expectations, and the fact
  that elements are unweighted.  For the second, let $\pi_{< e}$ be the
  outcomes of elements before $e$ in the ordering. Note that the event
  $\mathbf{1}(e \text{ probed})$ is completely determined by $\pi_{<
    e}$. Moreover,
  \[ \Pr[ e \text{ picked} \mid \pi_{< e} ] = \mathbf{1}\left(e \text{
    probed}\mid \pi_{< e}\right) \cdot p_e~. \]
  Hence, the expected value of the algorithm is
{\small  \begin{align*}
    \E[\alg] &= \sum_{e} \sum_{\pi_{< e}} \Pr[\pi_{< e}] \cdot \Pr[ e
    \text{ picked} \mid \pi_{< e} ] = \sum_{e} \sum_{\pi_{< e}}
    \Pr[\pi_{< e}] \cdot \mathbf{1}\left(e \text{ probed}\mid \pi_{< e}\right) \cdot p_e \\
    &= \sum_{e} \sum_{\pi}
    \Pr[\pi] \cdot \mathbf{1}\left(e \text{ probed}\mid \pi_{< e}\right) \cdot p_e 
    = \sum_{e} \sum_{\pi}
    \Pr[\pi] \cdot \sum_{e \in Q_{\pi}}  p_e.
  \end{align*}}%
  Above, we used the fact that $e$'s being probed (or equivalently, it's
  lying in $Q_{\pi}$) was purely a function of $\pi_{< e}$. And that $e$
  being active is independent of all others.
\end{proof}}

\begin{lemma}
  \label{lem:dual}
  For each outcome $\pi$, there is a feasible dual of value at most
  $k |S_\pi| + k' \sum_{e \in Q_\pi} p_e$. Moreover, there is a
  feasible dual of value at least $(k+k') \E[\alg]$.
\end{lemma}

The following proof is similar to that of Fisher et al.~\cite{FNW78}
showing that the greedy algorithm is a $k$-approximation for the
intersection of $k$ matroids.

\begin{proof}
  Let $A = \spn(S_\pi)$ be the span of the set of picked elements $S_\pi$; note that by Claim~\ref{cl:psys-prop}, $r(A)\le k\cdot |S_\pi|$. We set $\alpha(A) = 1$, and all other $\alpha$ variables to zero.

Let the set of probed elements $Q_\pi = \{a_1, a_2, \ldots, a_\ell\}$ in
  this order. Define
  \[ \beta(\spn'(\{a_1, a_2, \ldots, a_h\})) := p_{a_h} - p_{a_{h+1}}
  \geq 0 \] for all $h \in \{1, \ldots, \ell\}$ (where we imagine
  $p_{a_{\ell+1}} = 0$). This is also well-defined since every subset of
  $Q_\pi$ is independent in $\pot$. The non-negativity follows from the
  greedy algorithm that probes elements in decreasing probabilities. The dual objective value equals:
{\small $$r(A) + \sum_{h = 1}^\ell  r'(\spn'(\{a_1, a_2, \ldots, a_h\}))
  \cdot (p_{a_h} - p_{a_{h+1}}) \,\,\le \,\, k \cdot |S_\pi| + \sum_{h =
    1}^\ell k'\cdot h \cdot (p_{a_h} - p_{a_{h+1}}),$$}%
which is $k\cdot |S_\pi| + k' \sum_{e \in Q_\pi} p_e$. The inequality is by Claim~\ref{cl:psys-prop}. Next we show that the  dual solution is feasible. The non-negativity is clearly satisfied, so  it remains to check feasibility of the dual covering constraints.  For  any $e \in V$, 
  \begin{OneLiners}
  \item Case I: $e \in Q_\pi$. Say $e = a_g$ in the ordering of the set
    $Q_\pi$ . Then $e$ lies in $\spn'(\{a_1, a_2, \ldots, a_h\})$ for
    all $h \geq g$. Hence, the left hand side
    of $e$'s covering constraint contributes at least
    \[ \sum_{h = g}^\ell \beta(\spn'(\{a_1, a_2,
    \ldots, a_h\})) =   \sum_{h = g}^\ell (p_{a_h} -
    p_{a_{h+1}}) = p_{a_g} = p_e. \]
  \item Case II: $e \not\in Q_\pi$ because of the outer constraint. Say
    $e$ was seen when the $Q$ set was $\{a_1, a_2, \ldots, a_g\}$. Then $e \in
    \spn'(\{a_1, a_2, \ldots, a_h\})$ for all $h \geq g$. In this case,
    the left hand side contributes at least
    \[ \sum_{h = g}^\ell \beta(\spn'(\{a_1, a_2, \ldots, a_h\})) =
    \sum_{h = g}^\ell (p_{a_h} - p_{a_{h+1}}) = p_{a_g} \geq p_e. \]
    Here we used the fact that elements are considered in decreasing
    order of their probabilities.
  \item Case III: $e \not\in Q_\pi$ because of the inner constraint. Then $e \in
    \spn(S_\pi) = A$, and hence the $p_e \, \sum_{S: e \in S}
    \alpha(S) = p_e\, \alpha(A) = p_e$.
  \end{OneLiners}
  This proves the first part of the lemma. Taking expectations over
  $\pi$, the resulting convex combination $\sum_\pi \Pr[\pi]
  (\boldsymbol{\alpha}_\pi, \boldsymbol{\beta}_\pi)$ of these feasible
  duals is another feasible dual of value $k\, \E[|S_\pi|] + k'\, \E[ \sum_{e
    \in Q_\pi} p_e ]$, which by Lemma~\ref{lem:alg-value} equals
  $(k+k')\E[\alg]$. 
\end{proof}

Our analysis for the greedy algorithm is tight. In particular, if all
$p_e$'s equal one, and the inner and outer constraints are intersections
of (arbitrary) partition matroids, then we obtain the greedy algorithm
for $(\kin+\kot)$-dimensional matching. The approximation ratio in this
case is known to be exactly $\kin+\kot$. 

\stocoption{}{
\paragraph{Application to Unweighted Stochastic Matching.} When the
inner constraint is matching (which is a $2$-system) and the outer
constraint is $b$-matching (also a $2$-system) on the same graph, we
obtain the unweighted stochastic matching problem of Chen et
al.~\cite{CIKMR09}. Hence Theorem~\ref{thm:unit-wt} gives an alternate
proof of greedy being a $4$-approximation~\cite{CIKMR09}. 
We know now that greedy is a $2$-approximation~\cite{A11}, but we
currently do not know an LP-based proof of this bound.
}


\section{Weighted Stochastic Probing}
\label{sec:weighted}

We now turn to the general weighted case of stochastic probing. Here the
natural combinatorial algorithms perform poorly, so we use linear
programming relaxations of the problem, which we round to get
non-adaptive policies.  Given an instance of the stochastic probing
problem with inner constraints $(V,\pin)$ and outer constraints
$(V,\pot)$, we use the following LP relaxation:
\begin{align*}
\max ~~~& \textstyle \sum _{e\in V} w_e\cdot x_e & \\
s.t. ~~~& x_e = p_e\cdot y_e \qquad \forall e\in V &  ({\cal{LP}})\\
        & x\in \p(\pin)  &  \\
 & y\in \p(\pot)  &  
\end{align*}
We assume that the LP relaxations of the inner and outer constraints can
be solved efficiently: this is true for matroids, knapsacks, UFP on
trees, and their intersections. For general $k$-systems, it is not known
if this LP can be solved exactly. However, using the fact that the
greedy algorithm achieves a $\frac{1}{k}$-approximation for
maximizing linear objective functions over $k$-systems (even with
respect to the LP relaxation, which follows from~\cite{FNW78}, or the
proof of Lemma~\ref{lem:dual}), \stocoption{ and the equivalence of separation and
optimization, we can obtain a $\frac{1}{k}$-approximate LP solution when
\pin and \pot are $k$-systems.}{ and the equivalence of approximate separation and
optimization~\cite{J03}, we can obtain a $\frac{1}{\max\{\kin,\kot\}}$-approximate LP solution when \pin and \pot are arbitrary $\kin$ and $\kot$ systems. }

\begin{cl}\label{cl:lp-opt}
  The optimal value of (\lp) $\ge$ optimal value of the 
  probing instance.
\end{cl}
\stocoption{}{\begin{proof}
  Let $y^*_e$ denote the probability that element $e$ is probed by the
  optimal strategy; i.e., $y^*_e= \Pr[e\in Q^*]$. Also let $x^*_e$
  denote the probability that element $e$ is chosen in the final
  solution, $x^*_e= \Pr[e\in S^*]$. Due to the constraints, we have $Q^*
  \in \pot$ and $S^* \in \pin$, and hence $y^*\in \p(\pot)$ and $x^*\in
  \p(\pin)$. Moreover, 
  \[ x^*_e = \Pr[e\in S^*] = \Pr[e\in Q^* \mbox{ and }e \mbox{ active}]
  = p_e \cdot
  \Pr[e\in Q^*]=p_e\cdot y^*_e, \quad \forall e\in V.
  \]
  Here we used the fact that the probability of element $e$ being active
  is independent of the past decisions, and in particular, of the
  optimal strategy's decision to probe $e$.  Thus $(x^*,y^*)$ is a
  feasible solution to \lp. Finally, the optimal value of the probing
  problem instance is $\sum_{e\in V} w_e\cdot \Pr[e\in S^*] = \sum_e w_e
  x^*_e$, which is the
  LP objective value of $(x^*,y^*)$.
\end{proof}
}

\stocoption{}{\subsection{Contention-Resolution Schemes} \label{sec:cr-schemes}}

Given a solution $(x,y)$ for the LP relaxation, we need to get a policy
from it. Our rounding algorithm is based on the elegant abstraction of
{\em contention resolution schemes (CR schemes)}, as defined in Chekuri
et al.~\cite{CVZ11}. Here is the formal definition, and the main theorem
we will use.
\begin{definition}\label{defn:cr}
  An independence system $(V,{\cal J}\sse 2^V)$ with LP-relaxation
  $\p({\cal J})$ admits a monotone $(b,c)$ {\bf CR-scheme} if, for any $z\in \p({\cal
    J})$ there is a (possibly randomized) mapping $\pi: 2^V\rightarrow
  {\cal J}$ such that:
  \begin{OneLiners}
  \item[(i)] If $I\sse V$ is a random subset where each
    element $e\in V$ is chosen independently with probability $b\cdot x_e$, 
    $\Pr_{I,\pi}[e\in \pi(I) \mid e\in I]\ge c$ for all $e\in V$.
  \item[(ii)] For any $e\in I_1\sse I_2\sse V$, $\Pr_{\pi}[e\in \pi(I_1)] \ge
    \Pr_{\pi}[e\in \pi(I_2)]$.
  \item[(iii)] The map $\pi$ can be computed in polynomial time.
\end{OneLiners}
Moreover, $\pi:2^V\rightarrow {\cal J}$ is a $(b,c)$ {\bf ordered CR-scheme} if there is a (possibly random) permutation $\sigma$ on $V$ so that for each $I\sse V$, $\pi(I)$ is the maximal independent subset of $I$ obtained by considering elements in the order of $\sigma$.
\end{definition}

\begin{theorem}[\cite{CVZ11,CMS07,BKNS10,CHMS10}]\label{thm:cr-schemes}
There are monotone CR-schemes for the following 
independence systems 
(below, $0<b\le 1$ is any value unless specified otherwise) 
  \begin{OneLiners}
  \item  $(b,\,(1-e^{-b})/b)$ CR-scheme for matroids. 
  \item $(b, 1-k\cdot b)$ ordered CR-scheme for $k$-systems. 
  \item $(b, 1-6b)$ ordered CR-scheme for unsplittable flow on trees, with the ``no bottleneck'' assumption, for any $0<b\le 1/60$.
  \item $(b, 1-2kb)$ CR-scheme for $k$-column sparse packing integer programs. 
  \end{OneLiners}
\end{theorem}
\stocoption{}{The CR-scheme for $k$-systems can be inferred from Lemma~4.12 in~\cite{CVZ11} using the observation that $r(\spn(R))\le k\cdot |R|$ for any $R\sse V$ in a $k$-system.}

\stocoption{}{\subsection{How to Round the LP Solution} \label{sec:algo}}

Given the formalism of CR schemes, we can now state our main result for
rounding a solution to the relaxation (\lp).

\begin{theorem}\label{thm:stoc-probe}
  Consider any instance of the stochastic probing problem with
  \begin{OneLiners}
  \item[(i)] $(b,c_{out})$ CR-scheme for 
  $\p(\pot)$.
  \item[(ii)] Monotone $(b,c_{in})$ {\bf ordered} CR-scheme for 
     $\p(\pin)$. 
  \end{OneLiners}
  Then there is a $b\cdot \left(\,c_{out}+c_{in}-1\,\right)$-approximation
  algorithm for the weighted stochastic probing problem.
\end{theorem}

Before we prove Theorem~\ref{thm:stoc-probe}, we observe that combining
Theorems~\ref{thm:stoc-probe} and~\ref{thm:cr-schemes} gives us,
for example:
\begin{itemize}
\item a $1/(4(k+\ell))$-approximation algorithm when the inner and
  outer constraints are intersections of $k$ and $\ell$ matroids
  respectively. 
\item an $\Omega(1)$-approximation algorithm when the inner and outer
  constraints are unsplittable flows on trees/paths satisfying the
  no-bottleneck assumption.
\stocoption{}{\item an $\Omega\left(1/(k+\ell)^2\right)$-approximation algorithm when the inner and outer constraints are arbitrary $k$ and $\ell$ systems. Here, we lose an additional $k+\ell$ factor in solving \lp approximately.}
\end{itemize}

{\bf The Rounding Algorithm.} Let $\pi_{out}$ denote the randomized mapping corresponding to a
$(b,c_{out})$ CR-scheme for $y\in \p(\pot)$, and $\pi_{in}$ be that
corresponding to a $(b,c_{in})$ CR-scheme for $x\in \p(\pin)$. The
algorithm to round the LP solution $(x,y)$ for weighted stochastic
probing appears as Algorithm~\ref{alg:wt-probe}.

\begin{algorithm}[h!] \caption{Rounding Algorithm for Weighted Probing\label{alg:wt-probe}} 
  \begin{algorithmic}[1]
    \STATE   \label{item:1} Pick $I\sse 2^V$ by choosing each $e\in V$
  independently with probability $b\cdot y_e$.
    \STATE   \label{item:2} Let $P=\pi_{out}(I)$. {\small (By definition of the CR
  scheme, $P \in \pot$ with probability one.)}
    \STATE   \label{item:3} Order elements in $P$ according to $\sigma$ (the inner ordered CR scheme) to get $e_1,e_2,\ldots,e_{|P|}$.
    \STATE   \label{item:4} Set $S\gets \emptyset$.

    \FOR{$i=1,\ldots,|P|$}
    \IF{$(S\cup \{e_i\} \in \pin)$ }
 \STATE \label{item:5} Probe $e_i$: set $S\gets S\cup \{e_i\}$ if $e_i$ is active, and $S\gets S$ otherwise.
\ENDIF     
    \ENDFOR
  \end{algorithmic}
\end{algorithm}


{\bf The Analysis.} We now show that $\E[w(S)]$ is large compared to the LP value $\sum_e
w_e x_e$. To begin, a few observations about this algorithm. Note that
this is a randomized strategy, since there is randomness in the choice
of $I$ and maybe in the maps $\pi_{out}$ and $\pi_{in}$. Also, by the CR
scheme properties, the probed elements are in $\pot$, and the chosen
elements in $\pin$. Finally, having chosen the set $P$ to (potentially)
probe, the elements actually probed in step~\ref{item:5} 
relies on the {\em ordered} CR scheme for the inner constraints. 
\stocoption{}{In Appendix~\ref{sec:bad-eg} we show that some simpler rounding algorithms that work for stochastic matching do not apply in this more general setting.}

Recall that $I\sse V$ is the random set where each element $e$ is
included independently with probability $b\cdot y_e$; also
$P=\pi_{out}(I)$. Let $J\sse V$ be the set of active elements; i.e.,
each $e \in V$ is present in $J$ independently with probability $p_e$.
The set of chosen elements is now $S=\pi_{in}(P\cap J)$. The main lemma
is now:
\begin{lemma}
  \label{lem:main} For any $e\in V$,
  \[ \Pr_{I,\pi_{out},\,J,\, \pi_{in} } \,\, \left[ e\in
    \pi_{in}\left( \pi_{out}(I) \cap J\right) \right] \quad \ge \quad b\cdot
  (c_{out}+c_{in}-1) \cdot x_e, \] where $b, c_{out}, c_{in}$ are
  parameters given by our CR-schemes.
\end{lemma}
\begin{proof}
Recall that $P=\pi_{out}(I)$, so we want to lower bound: 
\begin{eqnarray}
\Pr[e\in \pi_{in}(P\cap J)] &=& \Pr[e\in \pi_{in}(P\cap J)\wedge e\in I\cap J\cap P] \notag\\
&=& \Pr[e\in I\cap J\cap P] - \Pr[e\not\in \pi_{in}(P\cap J)\wedge e\in I\cap J\cap P] \notag \\
&\ge & bx_e\cdot c_{out} - \Pr[e\not\in \pi_{in}(P\cap J)\wedge e\in I\cap J\cap P], \label{eq:main-lem1} 
\end{eqnarray}
where the inequality uses $\Pr[e\in I\cap J]=by_e\cdot p_e=bx_e$ and $\Pr[e\in P=\pi_{out}(I)|e\in I\cap J] \ge c_{out}$ by Definition~\ref{defn:cr}(i) applied to the outer CR
scheme, since $I$ is a random subset chosen according to $b\cdot y$
where $y\in \p(\pot)$. 

We now upper bound $\Pr[e\not\in \pi_{in}(P\cap J)\wedge e\in I\cap J\cap P]$ by $(1-c_{in})\cdot bx_e$ which combined with~\eqref{eq:main-lem1} would prove the lemma. Now, 
condition on any instantiation $I=I_1$, $P=\pi_{out}(I_1)=P_1\sse I_1$ and $J=J_1$ such that $e\in I_1\cap J_1\cap P_1$. Then,
\begin{equation}\label{eq:main-lem2}
\Pr[e\not\in \pi_{in}(P_1\cap J_1)] \quad \le \quad \Pr[e\not\in \pi_{in}(I_1\cap J_1)],
\end{equation} 
by Definition~\ref{defn:cr}(ii) applied to the inner CR scheme (since $e\in P_1\cap J_1\sse I_1\cap J_1$). Taking a linear combination of the inequalities in~\eqref{eq:main-lem2} with respective multipliers $\Pr[I=I_1,J=J_1,P=P_1]$ (where $e\in I_1\cap J_1\cap P_1$), we obtain
\begin{eqnarray*}
\Pr[e\not\in \pi_{in}(P\cap J)\wedge e\in I\cap J\cap P] & \le & \Pr[e\not\in \pi_{in}(I\cap J)\wedge e\in I\cap J\cap P]\\
&\le & \Pr[e\not\in \pi_{in}(I\cap J)\wedge e\in I\cap J]\\
&=& b x_e \cdot \Pr[e\not\in \pi_{in}(I\cap J) | e\in I\cap J]
\end{eqnarray*}
where the equality uses $\Pr[e\in I\cap J]=by_e\cdot p_e=bx_e$. The last expression above is at most $bx_e(1-c_{in})$ by Definition~\ref{defn:cr}(i) applied to the inner CR
scheme, since $I\cap J$ is a random subset chosen according to $b\cdot x$
where $x\in \p(\pin)$. This proves $\Pr[e\not\in \pi_{in}(P\cap J)\wedge e\in I\cap J\cap P]\le (1-c_{in})\cdot bx_e$ as desired.
\end{proof}


Consequently, the expected weight of the chosen set $S$ is
$$\E\left[ \sum_{e\in S} w_e\right]  \,\, =  \,\, \sum_{e\in V} w_e\cdot
\Pr\left[ e\in \pi_{in}\left( P\cap J\right) \right] \,\, \ge  \,\,
b(c_{in}+c_{out}-1)\cdot \sum_{e\in V} w_e\cdot x_e. $$
The inequality uses Lemma~\ref{lem:main}. This completes the proof of Theorem~\ref{thm:stoc-probe}.

\stocoption{}{
\medskip
\noindent
{\bf Remark:} We note that our results also hold in a slightly more general model where the elements are not necessarily independent, but every set $T\in \pot$ is mutually independent.\footnote{A set $\{E_i\}_{i=1}^\ell$ of events is mutually independent if for any subset $L\sse[\ell]$ we have $\Pr\left[\wedge_{i\in L} E_i\right]=\Pi_{i\in L} \Pr[E_i]$.} 
\begin{OneLiners}
\item Observe that $\lp$ is a valid relaxation for stochastic probing, even in this setting. The only change in the proof of Claim~\ref{cl:lp-opt} is: if $Q^*$ and $S^*$ denote the sets of probed and chosen elements in an optimal policy then $\Pr[e\in S^*] =\Pr[e\in Q^* \text{ and $e$ active}]=\Pr[e\in Q^*]\cdot \Pr[\text{$e$ active}\mid e\in Q^*]=\Pr[e\in Q^*]\cdot p_e$, where the last equality uses the fact that at any point in the optimal policy when $e$ is probed, ``element $e$ being active'' is independent of the previously observed elements (which along with $e$ is some set in $\pot$ and hence is mutually independent). 
\item Moreover, in Lemma~\ref{lem:main}, if we let $J\sse V$ denote the  random subset where each element $e$ is present independently with probability $p_e$ and $J_a\sse V$ the set of active elements, then $\Pr\left[e\in \pi_{in}\left( \pi_{out}(I) \cap J_a\right) \right]=\Pr\left[e\in \pi_{in}\left( \pi_{out}(I) \cap J\right) \right]$. This is because, conditioning on any $I=I_1$ and $P=\pi_{out}(I_1)=P_1$, the distributions of $J_a\cap P_1$ and $J\cap P_1$ are identical (by mutual independence of $P_1$).
\end{OneLiners}
\section{Bayesian Single Parameter Mechanism Design}\label{sec:spm}
\def\bb{\mathbf{b}}

In this section, we show how a Bayesian single-parameter auction problem
can be modeled as a stochastic probing problem, yielding new
posted-price mechanisms for such auctions. 

Formally, we consider a {\em Bayesian} mechanism design problem with one
seller and $n$ {\em single-parameter} agents that bid for service.  The
term ``single-parameter'' means that each agent $i$'s private
information is represented by a single number $v_i$, which is the
agent's {\em valuation}.  In the Bayesian setting, the valuation $v_i
\in \{0,1,\ldots,B\}$ is drawn from an independent probability
distribution $\D_i$.\footnote{We can also handle continuous
  distributions by approximating them via discrete distributions, at the
  loss of a small constant factor.} These valuations are private
knowledge, but the distributions $\D_i$ are publicly known. Each agent
$i\in[n]$ submits a {\em bid} $b_i\in\{0,1,\ldots,B\}$ representing his
valuation. The seller has a \emph{feasibility constraint} given by a
downward closed set system $\I\sse 2^{[n]}$, and hence can serve any set
of agents from \I.

A {\em mechanism} is a function that maps a bid-vector $\bb \in
\{0,1,\ldots,B\}^n$ to an {\em allocation} $A(\bb)\in\I$, along with
{\em prices} $\pi_i(\bb)$ to be paid by each agent $i\in[n]$. For
notational convenience, we define $X_i(\bb) := \mathbf{1}_{i\in A(\bb)}$
denoting whether or not agent $i$ receives service. The {\em utility} of
agent $i$ under bids $\bb$ is $v_i\cdot X_i(\bb)-\pi_i(\bb)$. We
consider mechanisms satisfying the following standard properties:
\begin{itemize}
\item {\em Voluntary participation:} an agent pays only when receiving service and the payment is at most his bid. For all bid vectors $\bb$, $\pi_i(\bb) \le X_i(\bb)\cdot b_i$ for all agents $i\in[n]$.
\item {\em No positive transfers:} the mechanism does not pay agents,  $\pi_i(\bb)\ge 0$ for all $i\in[n]$ and $\bb$.
\item {\em Truthful in expectation:} For each $i\in[n]$ and $v_i\in\{0,1,\ldots,B\}$, if $v_i$ is agent $i$'s true valuation then his expected utility by 
bidding $v_i$ is at least his expected utility under any other bid $b_i\in\{0,1,\ldots,B\}$, i.e.,
\end{itemize}
{\small $$\E_{b_j\gets \D_j : j\ne i}\,\left[ v_i\cdot X_i(\bb_{-i}, v_i) - \pi_i(\bb_{-i}, v_i) \right] \,\, \ge \,\, \E_{b_j\gets \D_j : j\ne i}\,\left[ v_i\cdot X_i(\bb_{-i}, b_i) - \pi_i(\bb_{-i}, b_i) \right]
$$}

We are interested in designing a mechanism that maximizes the expected
revenue. The well-known \emph{Myerson mechanism}~\cite{M81} is optimal
for the single-parameter setting, and it proceeds by reducing the
revenue maximization problem to the welfare-maximization setting (which
can then be solved using the VCG mechanism). However the resulting
mechanism can be complicated to implement, and is computationally hard
under combinatorial feasibility constraints such as intersections of
more than two matroids. Hence, simpler mechanisms such as ``sequential
posted price mechanisms'' (SPM) are often desirable in practice;
see~\cite{CHMS10,Y11,KW12} for further discussion on this.  In an SPM,
the seller offers ``take it or leave it'' prices to the agents
one-by-one. When the feasible set \I is given by the intersection of $k$
matroid constraints, Chawla et al.~\cite{CHMS10} showed the existence of
SPMs achieving a $\frac1{k+1}$-approximation to the optimal mechanism.
However, when the posted prices are to be computed in polynomial time,
the approximation ratio becomes $\Omega(1/k^2)$. Indeed, getting
$\Omega(1/k)$-approximate SPMs for intersections of $k \geq 2$ matroids was
an open problem before this work.

In this section, we show that approximately optimal SPMs can be obtained
as an application of the stochastic probing problem. Since our algorithm
for computing prices runs in polynomial time, we obtain a polynomial
time SPM for $k$-matroid constraints that is an
$\Omega(1/k)$-approximation to the optimal mechanism. We proceed by
first showing that computing the optimal posted price mechanism is an
instance of the matroid constrained probing problem. Then we show that
the LP relaxation of this probing problem (which we use for our
algorithm) has value at least that of the optimal (potentially
non-posted price) mechanism as well, which completes the argument.

\subsection{Posted Price Mechanism as Probing Problem}

Given the distributions $\D_i$ of each agent $i\in[n]$, and the
feasibility constraint \I, we are interested in computing a {\em
  sequential posted price} mechanism. This corresponds to setting prices
$\pi_i$ for each agent $i\in [n]$, and making ``take it or leave it''
offers to agents in a suitable sequence (while ensuring feasibility in
\I). A rational agent will accept an offer if and only if the posted
price is at most his valuation. It is clear that any such mechanism is
truthful since the prices are independent of bids. In fact this
mechanism is truthful even when each agent knows the precise bids of all
other agents, which is a stronger condition than truthfulness in
expectation.

Consider an instance of the stochastic probing problem with:
\begin{itemize}
\item Universe $V:=\{(i,c)\, :\, i\in[n],\, c\in\{0,1,\ldots,B\} \}$.
\item Weights $w_{i,c} = c$ for all $(i,c)\in V$.
\item Probabilities $p_{i,c} = \Pr_{v_i\gets \D_i}\left[v_i \ge c\right]$ for all $(i,c)\in V$.
\item The outer constraint being a partition matroid: $\pot$ consists of
  all subsets $S\sse V$ with $|S\cap \{(i,c)\}_{c=0}^B|\le 1$ for all
  $i\in[n]$. This corresponds to offering at most one price to each
  agent.
\item The inner constraint being the natural lifting of the seller's
  feasibility constraint (on universe $[n]$) to $V$, where
  $\{(i,c)\}_{c=0}^B$ are copies of $i$. Formally, $\pin$ consists of
  all subsets $S\sse V$ with (a) $|S\cap \{(i,c)\}_{c=0}^B|\le 1$ for
  all $i\in[n]$ and (b) $\{i\in[n] : \exists (i,c)\in S\}\in \I$.
\end{itemize} 
Notice that if \I is given by an intersection of $p$ matroids then so is
$\pin$.

Due to the outer constraint, a solution to this probing problem never
probes two copies of the same agent. This ensures two properties: (1)
the independence assumption on elements of $V$ agrees with the auction
setting where copies $\{(i,c)\}_{c=0}^B$ of each agent $i$ are actually
dependent, and (2) we obtain a posted price mechanism. Moreover, the
inner constraint handles the feasibility constraint \I. Thus, solutions
to this probing problem correspond precisely to sequential posted price
mechanisms and vice versa. Instead of modeling the elements $\{(i,c)\}_{c=0}^B$ of each agent $i$ as being active independently, we could also model their joint distribution induced by $\D_i$: since every set in $\pot$ is now mutually independent, Theorem~\ref{thm:stoc-probe} still applies (as noted in the end of Section~\ref{sec:weighted}).

We can now use our algorithm for the weighted stochastic probing problem
to obtain an approximately optimal SPM. In the next subsection, we show
that the optimal revenue (of any mechanism, which may potentially be
non-posted-price) is at most the value of the stochastic probing LP.
Since our approximation ratio for stochastic probing is relative to this
LP, we obtain the following result (setting $b=\frac{1}{2k+1}$, $c_{in}=1-kb$ and $c_{out}=(1-e^{-b})/b\ge 1-b/2$ in Theorem~\ref{thm:stoc-probe})
\begin{theorem}\label{thm:spm}
  There is a polynomial-time sequential posted price mechanism for $k$
  matroid intersection constraints, which has revenue at least
  $\frac1{4k+2}$ times the revenue of the optimal mechanism.
\end{theorem}
More generally, this result holds for any feasibility constraint $\I$ that admits an {\em ordered} CR scheme, where the approximation ratio depends on the quality of the CR scheme. For example, this also implies a constant factor approximate SPM when $\I$ is given by an unsplittable flow on trees.

\subsection{Bounding the Optimal Mechanism}

Recall that our algorithm for the weighted probing problem is based on the LP relaxation \lp. For instances  corresponding to the Bayesian mechanism design problem (from the reduction above), this LP is:
{\small \begin{align}
LP_P \quad =\quad \max & \qquad \sum_{i \in [n]} \sum_{c=0}^B  c \cdot x_{i,c} \label{lpP:obj}\\
  \text{subject to}& \qquad x_{i,c}  = \Pr[v_i\ge c]\cdot y_{i,c} \qquad \forall (i,c)\in V \\
&\qquad \left\{\sum_{c=0}^B x_{i,c} \, : \,i\in [n]\right\} \,\in \, \p(\I)  \label{lpP:inner}\\
&\qquad \sum_{c=0}^B y_{i,c}  \le 1  \qquad\qquad \forall i\in [n]\label{lpP:outer}\\
&\qquad \mathbf{y} \ge \mathbf{0}. \label{lpP:non-neg}
\end{align}}
Constraint~\eqref{lpP:inner} is the inner constraint which is a lifting
of $\I$, and~\eqref{lpP:outer} is the outer partition matroid
constraint. We will show that the optimal value of this LP is least the
value of the optimal mechanism for the Bayesian auction problem. To do
so, we want to write an LP relaxation for the optimal mechanism.
Consider any mechanism given by allocations
$\left\{X_i(\bb)\right\}_{i\in[n]}$ and prices
$\left\{\pi_i(\bb)\right\}_{i\in[n]}$ as functions of bids. For each $i\in[n]$ and $c\in\{0,1,\ldots,B\}$, define:
$$z_{i,c} := \E_{b_j\gets \D_j : j\ne i}\,\left[ X_i(\bb_{-i}, c)\right]
\quad \mbox{and} \quad q_{i,c} := \E_{b_j\gets \D_j : j\ne i}\,\left[
  \pi_i(\bb_{-i}, c)\right].$$ 
For each agent $i$ and value $c$, when agent $i$
bids $c$, $z_{i,c}$ is the probability that $i$ is served by the
mechanism and $q_{i,c}$ is the expected price that $i$ is charged (both
expectations are taken over valuations of all other agents $[n]\setminus
i$).

\begin{lemma}[Myerson~\cite{M81}, Archer and Tardos~\cite{AT01}]
  \label{lem:mech-prop}
  Any mechanism that satisfies truthfulness in expectation and voluntary
  participation has:
  \begin{OneLiners}
  \item[A.] $z_{i,c}$ is non-decreasing in $c$, for all $i\in[n]$.
  \item[B.] $q_{i,c}\le c\cdot z_{i,c} - \sum_{h=0}^{c-1} z_{i,h}$ for all
    $c\in\{0,1,\ldots,B\}$ and $i\in[n]$. 
  \end{OneLiners}
\end{lemma}
\begin{proof}
  We provide a proof for completeness. Fix any agent $i\in[n]$.
  For the first property, we will show that $z_{i,c_1}\le z_{i,c_2}$ for
  any values $c_1<c_2$. By the truthfulness condition when $i$'s true
  valuation is $c_1$ and he bids $c_2$,
  $$c_1\cdot z_{i,c_1} - q_{i,c_1} \quad \ge \quad c_1\cdot z_{i,c_2} - q_{i,c_2}.$$ 
  Similarly when $i$'s true valuation is $c_2$ and he bids $c_1$,
  $$c_2\cdot z_{i,c_2} - q_{i,c_2} \quad \ge \quad c_2\cdot z_{i,c_1} - q_{i,c_1}.$$ 
  Adding the above two inequalities and rearranging, we get
  $(c_2-c_1)(z_{i,c_2}-z_{i,c_1})\ge 0$, i.e., $z_{i,c_2}\ge z_{i,c_1}$
  as desired. This proves the monotonicity of $z_{i,*}$.

  For the second property, fix also any value $c$. For each $h\le c$,
  when $i$'s true valuation is $h$ and he bids $h-1$, by truthfulness:
  $$q_{i,h} - q_{i,h-1} \quad \le \quad h\cdot z_{i,h} - h\cdot z_{i,h-1}.$$
  Adding this inequality over all $h\in \{1,\ldots,c\}$,
  $$q_{i,c}- q_{i,0} \quad \le \quad c\cdot z_{i,c} - \sum_{h=0}^{c-1} z_{i,h}.$$
  Now, voluntary participation implies that $q_{i,0}=0$, which proves
  the desired inequality.
\end{proof}

Also define $x_i := \sum_{c=0}^B \Pr[v_i=c]\cdot z_{i,c}$ for each $i\in [n]$. This denotes the probability that agent $i$ is served by the mechanism when all agents bid their true valuation. 

\begin{cl}\label{claim:feas-alloc}
  $\left\{x_i\,:\, i\in[n]\right\} \in \p(\I)$.
\end{cl}
\begin{proof}
  The feasibility constraint imposed by \I implies that $\left\{X_i(\bb)
    \,:\, i\in[n]\right\} \in \I$ for each $\bb\in \{0,1,\ldots,B\}^n$.
  Since $\p(\I)$ is a relaxation of \I, it is clear that
  $\left\{X_i(\bb) \,:\, i\in[n]\right\} \in \p(\I)$ for all $\bb$.
  Since $\p(\I)$ is a convex set, we have $\sum_{\bb} \lambda(\bb)\cdot
  \{X_i(\bb)\} \in \p(\I)$ for any convex multipliers $\lambda$s.
  Setting $\lambda(\bb):=\Pr\left[v_j=b_j\, \forall j\in [n]\right] =
  \Pi_{i=j}^n \Pr[v_j=b_j]$, for each $i\in[n]$ we have $\sum_{\bb} \lambda(\bb)\cdot X_i(\bb)$ equal to
  $$=\,\, \sum_{c=0}^B
  \Pr[v_i=c] \cdot \sum_{\bb_{-i}} \Pr[v_j=b_j : j\ne i] \cdot
  X_i(\bb_{-i},c) \,\, = \,\, \sum_{c=0}^B \Pr[v_i=c] \cdot z_{i,c}
  \,\,=\,\, x_i.$$ 
  Thus we obtain $\{x_i\}\in \p(\I)$ as claimed.
\end{proof}

Combining Lemma~\ref{lem:mech-prop} and Claim~\ref{claim:feas-alloc} we obtain the following LP relaxation for valid mechanisms:
\begin{align}
LP_M \quad =\quad  \max & \qquad \sum_{i \in [n]} \sum_{c=0}^B  \Pr[v_i=c] \cdot \left[ c\cdot z_{i,c} - \sum_{h=0}^{c-1} z_{i,h}\right] \label{lpM:obj}\\
  \text{subject to}& \qquad 0\le z_{i,0}\le z_{i,1}\le\cdots \le z_{i,B}\le 1 \qquad \forall i \in [n]  \label{lpM:monotone} \\
&\qquad x_i = \sum_{c=0}^B \Pr[v_i=c]\cdot z_{i,c} \qquad \qquad \forall i\in [n]\\
&\qquad \left\{x_i\,:\, i\in[n]\right\} \in \p(\I).\label{lpM:alloc}
\end{align}

We note that this LP is, in fact, a {\em relaxation} of the optimal
mechanism, and the optimal value of this LP may be strictly larger than
that of the optimal mechanism. (E.g.  for a single matroid constraint,
this gap can be as large as $\frac{e}{e-1}$.) It is also known that this
gap can be closed by adding exponentially many valid inequalities---the
so-called \emph{Border inequalities}. However, this is not important for
the current development, and the interested reader may refer
to~\cite{V11} for a thorough treatment of this area.

We are now ready to relate the above two LPs: $LP_M$, which is a
relaxation of the optimal mechanism, and $LP_P$, our relaxation of the
stochastic probing instance.
\begin{lemma}
$LP_P \,\,\ge \,\,LP_M$. Hence the optimal LP value of the stochastic probing instance is at least the revenue of the optimal mechanism.
\end{lemma}
\begin{proof}
Given any feasible solution $\langle z_{i,c}, x_i\rangle$ to $LP_M$, we construct a feasible solution $\langle x_{i,c}, y_{i,c}\rangle$ to $LP_P$ of the same objective value. Set $y_{i,c} := z_{i,c}-z_{i,c-1}$ for all $i\in [n]$ and $c\in\{0,1,\ldots,B\}$; using $z_{i,-1}=0$. Note that $\textbf{y}\ge 0$ due to constraint~\eqref{lpM:monotone}. Also $\sum_{c=0}^B y_{i,c} = z_{i,B} \le 1$ for each $i\in[n]$. This shows that constraints~\eqref{lpP:outer}-\eqref{lpP:non-neg} in $LP_P$ are satisfied.

Since $x_{i,c} = \Pr[v_i\ge c]\cdot y_{i,c} = \Pr[v_i\ge c]\cdot (z_{i,c}-z_{i,c-1})$, we have for each $i\in[n]$, 
$$\sum_{c=0}^B x_{i,c} = \sum_{c=0}^B \Pr[v_i\ge c]\cdot \left(z_{i,c}-z_{i,c-1}\right) = \sum_{c=0}^B z_{i,c} \cdot \left( \Pr[v_i\ge c] - \Pr[v_i\ge c+1] \right),$$ which equals $\sum_{c=0}^B z_{i,c} \cdot \Pr[v_i = c] = x_i$. Thus constraint~\eqref{lpM:alloc} in $LP_M$ implies constraint~\eqref{lpP:inner} in $LP_P$.  Finally, the objective value~\eqref{lpP:obj} of $LP_P$ is:
{\small \begin{align*}
\sum_{i \in [n]} \sum_{c=0}^B  c \cdot x_{i,c} & = \sum_{i \in [n]} \sum_{c=0}^B  c \cdot \Pr[v_i\ge c]\cdot \left(z_{i,c}-z_{i,c-1}\right)\\
& = \sum_{i \in [n]} \sum_{c=0}^B  z_{i,c}\cdot \left( c\cdot \Pr[v_i\ge c] - (c+1)\cdot \Pr[v_i\ge c+1] \right)\\
& = \sum_{i \in [n]} \sum_{c=0}^B  z_{i,c} \cdot \left( c\cdot \Pr[v_i = c] - \Pr[v_i\ge c+1]\right)\\
&= \sum_{i \in [n]} \sum_{c=0}^B  \Pr[v_i = c]\cdot c\cdot z_{i,c} \,\, - \,\, 
\sum_{i \in [n]} \sum_{c=0}^B  z_{i,c} \cdot \sum_{h=c+1}^B \Pr[v_i=h]\\
&= \sum_{i \in [n]} \sum_{c=0}^B  \Pr[v_i = c]\cdot c\cdot z_{i,c} \,\, - \,\, \sum_{i \in [n]} \sum_{c=0}^B  \Pr[v_i = c]\cdot \sum_{h=0}^{c-1} z_{i,h},
\end{align*}} 
which is exactly the objective~\eqref{lpM:obj}  of $LP_M$.
\end{proof}

\section{Unweighted Probing with Deadlines}
\label{sec:unwt-deadline}

In this section we consider a generalization of the stochastic probing
problem in the presence of  {\em global time}. Each probe requires one
unit of time and each element $e\in V$ has a {\em deadline} $d_e$ (in
the global time) after which it expires.
As before, we have inner \pin and outer \pot packing constraints on the set of chosen and probed elements respectively. We show that a natural greedy algorithm achieves a good approximation  for {\em unit-weighted stochastic probing with deadlines}, when the inner and outer constraints are $k$-systems. 

\begin{theorem}
  \label{thm:unwt-deadline}
  There is a $\frac{1}{2(\kin+\kot+1)}$-approximation algorithm for
  unweighted stochastic probing with deadlines, when \pin and \pot are
  $\kin$- and $\kot$-systems.
\end{theorem}

The main idea is to relax the global deadline constraints into an {\em
  outer laminar matroid} constraint $\cL$, and then relate the deadline
probing problem to the usual probing problem with outer constraints
$\pot\cap \cL$ and inner constraints \pin. The laminar matroid \cL is defined as follows. 
$$\cL \quad := \quad \left\{ U\sse V \,: \,|U\cap \{e: d_e\le t\}|\le t, \,\,\forall t\ge 1\right\}$$
Notice that the sets $D_t=\{e: d_e\le t\}$ for $t\ge 1$ form a {\em chain} family\footnote{A chain family is a collection of subsets $D_1\sse D_2\sse\cdots D_n$.}, and so \cL is indeed a laminar matroid. 

\begin{algorithm}[h!] \caption{Greedy Algorithm for Unweighted Probing with Deadlines} 
  \begin{algorithmic}[1]
    \STATE   $Q \gets \emptyset$, $S \gets \emptyset$, $B\gets \emptyset$ and $t\gets 1$.
    \FOR{$e$ in non-increasing order of $p_e$ value}
    \IF{$Q \cup \{e\} \in \pot\cap\cL$}
    \IF{$S \cup \{e\} \in \pin$}
    \STATE $Q \gets Q \cup \{e\}$ (i.e., potentially probe $e$)
    \IF{$t\le d_e$}
    \STATE probe element $e$, and $t\gets t+1$.
    \IF{$e$ active (happens with probability $p_e$)}
    \STATE $S \gets S \cup \{e\}$ (i.e., pick $e$)
    \ENDIF
	\ELSE
	\STATE $B\gets B\cup \{e\}$.
	\STATE $S\gets S\cup \{e\}$ with probability $p_e$; and $S\gets S$ otherwise.   
     \ENDIF    
     \ENDIF
    \ENDIF
    \ENDFOR
  \end{algorithmic}
  \label{alg:greedy-deadline}
\end{algorithm}

The variable $t$ in Algorithm~\ref{alg:greedy-deadline} tracks the global time, which increases by one after each probe. The probed elements are are $Q\setminus B$ and the chosen elements are $S\setminus B$. Observe that the algorithm defines a feasible policy since elements are only probed before their respective deadlines, and both inner \pin and outer \pot constraints are satisfied. We use the sets $Q,S,B$ to couple:
\begin{OneLiners}
\item this algorithm for the deadline probing instance \J, and 
\item the greedy algorithm (Section~\ref{sec:unwtd}) for the usual probing instance \K having inner constraints \pin and outer constraints $\pot\cap \cL$. 
\end{OneLiners}
Clearly, any feasible policy for \J is also feasible for \K.  Note that the greedy algorithm for \K will probe elements $Q$ and choose elements $S$. This is the reason why sets $Q$ and $S$ are updated even when a probe does not occur in \J. Also, $B$ denotes the set of elements that are probed in \K but not in \J. Consider a decision path $\pi$ down the decision trees associated with \J and \K; note that we couple instantiations in the two decision trees. By the analysis in Section~\ref{sec:unwtd}, the algorithm's objective in \K is 
$$\E[\alg(\K)]\quad =\quad \sum_\pi \Pr(\pi)\cdot \sum_{e\in Q_\pi} p_e \quad \ge \quad \frac{\opt(\K)}{\kin+\kot+1}.$$
Recall that \K has an inner $\kin$-system \pin and outer $(\kot+1)$-system $\pot\cap\cL$. Moreover, the algorithm's objective in \J is 
$$\E[\alg(\J)]\quad =\quad \sum_\pi \Pr(\pi)\cdot \sum_{e\in Q_\pi\setminus B_\pi} p_e.$$ 
The next lemma relates these two quantities.
\begin{lemma}
For any outcome $\pi$, $\sum_{e\in Q_\pi} p_e  \le 2\cdot \sum_{e\in Q_\pi\setminus B_\pi} p_e$.
\end{lemma}
\begin{proof}
This proof also relies crucially on the greedy ordering in terms of probabilities. Note that each element $e\in B_\pi$ must have been considered at time $t>d_e$. Moreover, time $t$ is increased only by elements $Q_\pi\setminus B_\pi$. Now by the greedy ordering, 
\begin{equation}\label{eq:deadline-1}
|(Q_\pi\setminus B_\pi) \cap \{f: p_f\ge p_e\}|\quad \ge \quad d_e,\quad \forall e\in B_\pi.
\end{equation}
Furthermore, since $Q_\pi\in \cL$ we also have $B_\pi\sse Q_\pi\in \cL$. So,
\begin{equation}\label{eq:deadline-2}
|B_\pi \cap \{f: d_f\le d_e\}|\quad \le \quad d_e,\quad \forall e\in B_\pi.
\end{equation}
Consider a bipartite graph $H$ with left vertices $B_\pi$ and right vertices $Q_\pi\setminus B_\pi$, with an edge between $e\in B_\pi$ and $f\in Q_\pi\setminus B_\pi$ iff $p_e\le p_f$. We claim that there is a left-saturating matching in $H$. It suffices to show Hall's condition  that for any subset $R\sse B_\pi$, its neighborhood $|\Gamma(R)|\ge |R|$. Let $e:=\arg\max\{d_g: g\in R\}$. Then, we have $|R|\le d_e$ using~\eqref{eq:deadline-2}, and  $|\Gamma(R)|\ge |\Gamma(e)|\ge d_e$ by~\eqref{eq:deadline-1}. Since graph $H$ has a left-saturating matching, it is clear that $\sum_{e\in B_\pi} p_e\le \sum_{e\in Q_\pi\setminus B_\pi} p_e$. 
\end{proof}

Using this lemma, and the above bounds for $\alg(\J)$ and $\alg(\K)$,
$$\alg(\J) \quad \ge \quad \frac{\opt(\K)}{2(\kin+\kot+1)} \quad \ge \quad \frac{\opt(\J)}{2(\kin+\kot+1)}.$$
The last inequality uses the fact that instance \K is a relaxation of instance \J. This proves the first part of Theorem~\ref{thm:unwt-deadline}.

\smallskip {\bf Stochastic matching with deadlines.} We give an
application of Theorem~\ref{thm:unwt-deadline} in the kidney exchange
setting. Consider a set of patients in a hospital, where each patient
$j$ is expected to be in the system for $d_j$ days. For each pair $i,j$
of patients, there is a probability $p_{i,j}$ of having a successful
match. On each day, the hospital can perform one compatibility test and
surgery between some pair of patients. If patient $j$ is not matched by
day $d_j$, he/she is assumed to have left the system. The difference from
the usual stochastic matching~\cite{CIKMR09,BGLMNR12} is that the
``timeout level'' of each patient decreases every day, irrespective of
whether he is probed. The goal is to schedule tests so as to maximize
the expected number of matched patients. This can be modeled as the
probing problem with deadlines, on groundset $V$ being the edges of the
complete graph on patients. Each edge $(i,j)$ has deadline
$\min\{d_i,d_j\}$ and probability $p_{i,j}$. There is no outer
constraint, and the inner constraint requires the chosen edges to form a
matching (2-system). Hence Theorem~\ref{thm:unwt-deadline} implies a
$\frac16$-approximation algorithm for this problem.

We note that an LP based approach as in~\cite{BGLMNR12} can also be used to obtain an approximation ratio of $1/6$ for this problem. However the above greedy algorithm is much simpler and extends to general $k$-system constraints. For the weighted case, our result (Theorem~\ref{thm:stoc-probe}) does not seem to extend directly to this setting of deadlines. We leave this as an open question.
}


\stocoption{
{\bf Acknowledgments:}
We thank Shuchi Chawla, Bobby Kleinberg, Tim Roughgarden, Rakesh Vohra,
and Matt Weinberg for helpful clarifications and discussions. We also thank an anonymous reviewer for pointing out that the LP for weighted stochastic probing can be solved approximately for general $k$-systems. Part of
this work was done when the first-named author was visiting the IEOR
Department at Columbia University, and IBM Thomas J.\ Watson Research
Center; he thanks them for their generous hospitality.}{
\section*{Acknowledgments}
We thank Shuchi Chawla, Bobby Kleinberg, Tim Roughgarden, Rakesh Vohra,
and Matt Weinberg for helpful clarifications and discussions. We also thank an anonymous reviewer for pointing out that the LP for weighted stochastic probing can be solved approximately for general $k$-systems, which leads to an approximation algorithm for weighted probing under $k$-system constraints. Part of
this work was done when the first-named author was visiting the IEOR
Department at Columbia University, and IBM Thomas J.\ Watson Research
Center; he thanks them for their generous hospitality.}


\stocoption{
\bibliographystyle{splncs}
\bibliography{probe}
}{
\bibliographystyle{plain}
\bibliography{probe}
}

\stocoption{}{

\appendix


\section{Bad Examples for Simpler LP-Rounding Algorithms}\label{sec:bad-eg}

Here we observe that some natural LP-rounding algorithms that work for stochastic matchings~\cite{BGLMNR12} do not work in the setting of general matroids. Let $(x,y)$ denote a solution to the linear relaxation \lp. Consider rounding this solution by considering elements to probe in the following order, where each element $e$ is probed with probability $b\cdot y_e$ when permitted by the inner and outer constraints ($0<b\le 1$ is some constant). 
\begin{OneLiners}
\item {\em Decreasing $w_e$ value.} There is no inner constraint, and the outer constraint is a graphic matroid on the graph $G$ (see Figure~\ref{fig:bad-ex}) consisting of edges $E:=\{e_i\}_{i=1}^n\bigcup \{f_i\}_{i=1}^n$ and $g$. The weights on edges $E$ are $M\gg 1$ each, and $w(g)=1$. The probabilities on edges $E$ are $\epsilon \ll \frac1{nM}$ each, and $p(g)=1$. The fractional solution $y$ has  value one on edge $g$ and value $1/2$ on each of $E$; the LP objective is at least one.  The expected weight from $E$ is at most $2nM\epsilon$. Since edge $g$ appears last in this order, the probability that $g$ is {\em not} blocked by the outer graphic matroid is at most $(1-b^2/4)^n$. (Note that if any edge of $E$ is blocked then so is $g$.) So the expected total weight of the rounding algorithm is at most $2nM\epsilon + (1-b^2/4)^n \ll 1$. 
\item {\em Decreasing $p_e$ value.} Again, there is no inner constraint and the outer constraint is a graphic matroid on $G$ (see Figure~\ref{fig:bad-ex}). The weights on $E$ are one, and $w(g)=L\gg n$. The probabilities on $E$ are one, and $p(g)=1/2$. $y$ has  value one on edge $g$ and value $1/2$ on each of $E$; so the LP objective is at least $L/2$. The expected weight from $E$ is at most $2n\ll L$. As before, since edge $g$ appears last in this order, the expected weight from $g$ is  at most $L\cdot (1-b^2/4)^n\ll L$. Hence the expected total weight is $\ll L$.
\item {\em Decreasing $w_e\cdot p_e$ value.} There is no outer constraint and the inner constraint is a graphic  matroid on graph $H$ (see Figure~\ref{fig:bad-ex}), which consists of edges $E:=\{e_i\}_{i=1}^n\bigcup \{f_i\}_{i=1}^n$ and $E':=\{g_j\}_{j=1}^N$. We set $N=n^2$. The weights are two on $E$, and $N$ on $E'$. The probabilities are $1/3$ on $E$, and $\frac1{3N}$ on $E'$.  $y$ has value one on all edges, and the LP objective is at least $N/3$. This order puts edges of $E$ before edges of $E'$. The probability that any particular edge of $E'$ is {\em not} blocked in the inner graphic matroid is at most $(1-b^2/9)^n<<1$. Thus the expected weight from $E'$ is at most $N^2\cdot (1-b^2/9)^n\cdot \frac1{3N}\ll N$. The expected weight from $E$ is at most $4n\ll N$. So the expected total weight is again much lesser than the LP objective.
\end{OneLiners}

\begin{figure}[h]
\begin{center}
\includegraphics[scale=0.75]{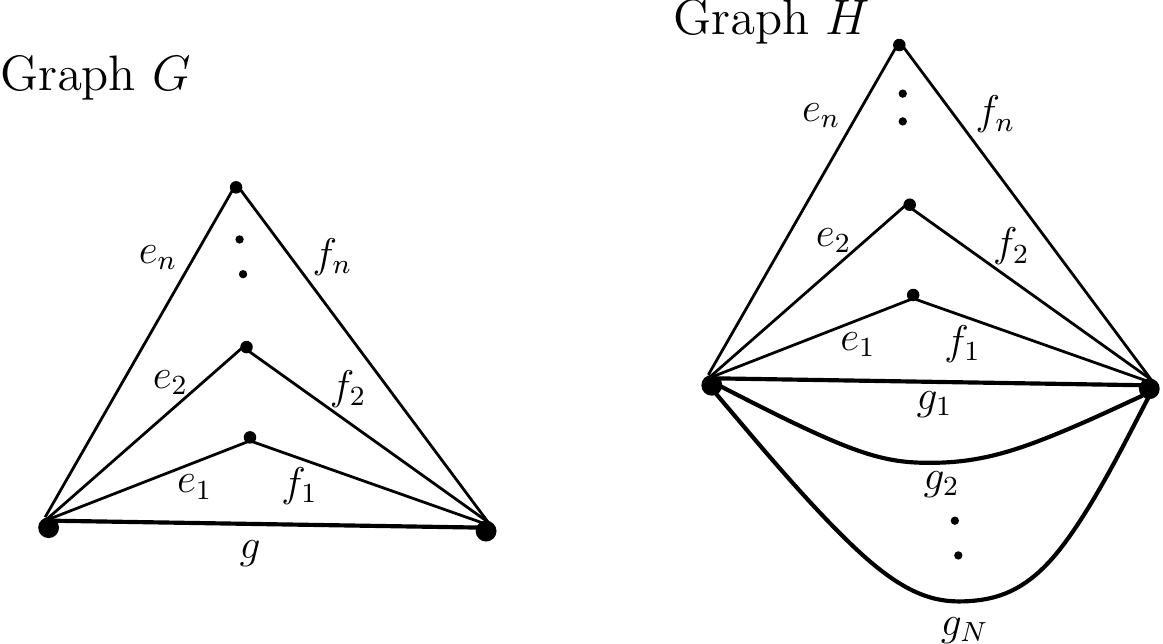}
\caption{Graphic matroids on $G$ and $H$.\label{fig:bad-ex}}
\end{center}
\end{figure}

}
\end{document}